\documentclass[lettersize,journal]{IEEEtran}
\usepackage{cite}
\usepackage{amsmath,amssymb,amsfonts}
\usepackage{amsthm}
\usepackage{algorithmic}
\usepackage{algorithm}
\usepackage{graphicx,color}
\usepackage{textcomp}
\usepackage{array}
\usepackage[caption=false,font=normalsize,labelfont=sf,textfont=sf]{subfig}
\usepackage{stfloats}
\usepackage{url}
\usepackage{verbatim}

\usepackage{enumitem}
\usepackage{booktabs}
\usepackage{multirow}
\usepackage{makecell}

\newtheorem{theorem}{Theorem}
\newtheorem{lemma}{Lemma}
\newtheorem{corollary}{Corollary}
\theoremstyle{definition}
\newtheorem{definition}{Definition}

\global
\global
\usepackage{bm}

\begin{document}

\title{Evict Early, Decode Efficiently: Lightweight Decoding Strategies for Low-Delay Sliding Window RLNC}
\author{Foteini~Karetsi,~\IEEEmembership{Graduate Student Member,~IEEE} and~Evangelos~Papapetrou,~\IEEEmembership{Member,~IEEE}
\thanks{F. Karetsi and E. Papapetrou are with the Department of Computer Science and Engineering, University of Ioannina, Ioannina, Greece; e-mail: f.karetsi, epap@uoi.gr}
}

\maketitle

\begin{abstract}
Random Linear Network Coding (RLNC), and in particular Sliding Window RLNC (SW-RLNC), is a promising approach for reliable, low-latency communication in dynamic environments such as 6G and non-terrestrial networks. While prior work has largely focused on encoder-side parameter tuning, such as the code rate and coding window size, decoder-side design remains comparatively underexplored despite its direct impact on end-to-end performance. This paper addresses decoding window management, namely the policies governing window-boundary updates, removal of obsolete packets, and delivery of recovered packets to upper layers, and shows that these policies critically shape the trade-off between computational complexity and decoding performance, particularly in-order delivery delay. We derive sufficient conditions under which packets can be safely removed from the decoding window without compromising decodability or violating in-order delivery guarantees, and, building on these conditions, propose two lightweight decoding window management schemes. Through extensive simulation-based evaluation across a range of channel conditions and code rates, we show that both schemes preserve decoding efficiency while substantially reducing decoding complexity, and that one additionally achieves a marked reduction in in-order delivery delay. These results establish decoding window management as a key design dimension for practical SW-RLNC systems.
\end{abstract}


\section{INTRODUCTION}
\IEEEPARstart{R}{andom} Linear Network Coding (RLNC)~\cite{ho2006random} has emerged as a powerful paradigm to enhance reliability and delay performance in modern communication systems, offering an alternative to conventional retransmission-based recovery and forward error correction (FEC) schemes~\cite{karetsi2022_lightweight,michel2024quirl,cohen2020adaptive}. RLNC decouples recovery from the identity of any specific lost packet, yielding a favorable delivery-delay trade-off: unlike traditional FEC schemes, its redundancy can be tailored to the prevailing channel conditions, avoiding both the overhead of over-provisioned protection and the recovery delay of under-provisioned protection, thereby improving goodput without sacrificing reliability. RLNC also remains effective under rapidly time-varying channel conditions, where the loss pattern cannot be anticipated in advance, and static protection schemes are consequently forced to operate conservatively. These properties are indispensable for demanding 6G use cases~\cite{3gpp22.261}, such as Ultra-Reliable Low-Latency Communications (URLLC)~\cite{3gpp38.913}, industrial automation~\cite{3gpp22.104}, and immersive extended reality (XR)~\cite{3gpp26.928}, where sub-millisecond latency budgets preclude feedback-driven retransmissions. Moreover, in frameworks governed by fine-grained resource management (e.g., network slicing~\cite{dos2019network}), RLNC mitigates the risk of over-provisioned redundancy competing for scarce radio resources. These advantages extend naturally to satellite and non-terrestrial networks, where long round-trip times, link asymmetry, and rapid channel degradation severely impair conventional recovery mechanisms. Consequently, RLNC is exceptionally well-suited for broadband access over low Earth orbit (LEO) constellations, delay-sensitive backhauling, and data delivery over intermittently connected or disrupted links.

Among RLNC variants, Sliding Window RLNC (SW-RLNC)~\cite{karzand2017design,wunderlich2017_caterpillar} is particularly attractive, as it enables continuous and equal protection of in-flight packets resulting in a favorable trade-off between performance and coding efficiency. This, however, comes at the cost of increased design complexity, since SW-RLNC requires careful multi-parameter tuning (e.g., coding window size, code rate and mode, decoding window size and management), with the optimal operating point being strongly dependent on channel dynamics. In recent years, prior work has primarily examined the optimization of individual SW-RLNC parameters, including the code rate~\cite{cohen2020adaptive,thai2014joint,zhang2021learning,wang2026adaptive,zhang2020low,malak2019tiny,cohen2022deepnp} and the coding window size~\cite{song2025adaptive,wowmom2023_karetsi,ali2023learning,li2023low}, as well as the joint optimization of parameter pairs such as coding window size and code rate~\cite{karetsi2026bayesian}. Nevertheless, the literature remains heavily skewed toward encoding-side optimization. In contrast, studies that explicitly address the decoding process are very limited in number, despite its critical role in end-to-end performance. Existing decoding-focused efforts either examine the design of on-the-fly decoding algorithms~\cite{icommlet2009_bioglio,iccw2009_heide,networking2011_heide} or assess the impact of the decoding window size on performance~\cite{wunderlich2017_caterpillar,ew2024_karetsi,roca2017less,wunderlich2017_wireless}.

In this work, we focus on the decoding side, with particular emphasis on the policy for managing the decoding window. Some discussion on decoding management already exists in the literature~\cite{wunderlich2017_wireless,wunderlich2017_caterpillar,ew2024_karetsi} but no specific policy has been proposed and/or evaluated. In general, a decoding window management policy comprises decisions regarding when to adjust window boundaries, remove obsolete or redundant packets, and deliver recovered packets to upper protocol layers. We argue that these decisions, and especially packet eviction policies, are central to the inherent trade-off between computational complexity and decoding performance, particularly with respect to in-order delivery delay. 
More specifically:
\begin{itemize}[leftmargin=1.2em, itemindent=0pt]
        \item We identify unnecessary packet retention in existing SW-RLNC decoder implementations and derive three sufficient conditions under which packets can be safely discarded from the decoding window.
        \item We prove that these conditions preserve not only the decoding correctness, but also the in-order delivery of packets.
        \item We propose two novel decoding window management strategies based on the identified conditions for packet eviction. Both strategies reduce the computational complexity of decoding, while the second also reduces in-order delivery delay. They are both designed to be compatible with existing on-the-fly decoding algorithms that keep the decoding matrix in reduced row echelon form.
        \item We experimentally demonstrate the preservation of decoding capability and quantify the benefits in computational complexity and in-order delay of the proposed strategies over an extensive range of configurations. We then show how these benefits can be leveraged to scale performance or reduce resource requirements.
\end{itemize}

In the remainder of this paper, we first present the necessary background on SW-RLNC, followed by a detailed review of prior work on SW-RLNC fine-tuning (Section~\ref{preliminaries}). We then introduce the conditions for early packet eviction and delineate the proposed decoding window management strategies (Section~\ref{dynamic-matrix-management}). In Section~\ref{evaluation}, we provide a comprehensive performance evaluation through extensive simulations. Finally, in Section~\ref{conclusion}, we summarize our main findings and outline directions for future research.

\section{BACKGROUND AND RELATED WORK}\label{preliminaries}
Maximizing the operational efficiency of 
SW-RLNC frameworks necessitates the meticulous calibration of a multi-dimensional parameter space distributed across both the transmitter (encoder) and the receiver (decoder). Within this paradigm, a complex interplay of architectural variables collectively dictates the system's performance. 
Consequently, an extensive body of literature has been dedicated to characterizing and optimizing these configurations across both conventional block-based and SW-RLNC architectures. Despite these extensive contributions, a critical asymmetry persists in the current research landscape. Existing optimization frameworks predominantly favor transmitter-side adaptations, often operating under the reductive assumption that decoder calibration is either trivial or secondary to encoding policies.

\subsection{Encoder Side Fine-Tuning}\label{encoder-side-literature}
 
The transmitter operates in recurring intervals, referred to as \emph{coding cycles}. During each cycle, the sender transmits a mixture of $k$ \textit{native} (or \textit{source}) packets (i.e., uncoded original packets) and $n$ \textit{coded} packets. The ratio of native to coded packets determines the code rate, $R=\frac{k}{k+n}$. This transmission mode is known as \emph{systematic} coding~\cite{iccw2009_heide}. Alternatively, the sender may adopt a \emph{full-vector} mode, in which all transmitted packets are coded, i.e., the $k$ native packets are replaced by $k$ coded packets. Coded packets are generated via linear combinations of the native packets currently residing within the \emph{coding window}, denoted by $\mathcal{C}_{w}$. The coding window is updated dynamically according to a First-In-First-Out (FIFO) management policy. Upon generation of a new native packet, it is appended to $\mathcal{C}_{w}$, and the chronologically oldest packet is evicted in order to maintain a maximum size, known as the \emph{coding window size} ($CW$). $CW$ is typically parameterized relative to the coding cycles as $CW = d \cdot k$, where $d \in \mathbb{Z}^{+}$ represents the \emph{coding depth}~\cite{karetsi2022_lightweight,wowmom2023_karetsi}. 

The relative merits of systematic~\cite{zhuo2026deep,claudel2025linc,dilanchian2025adjustable,huang2025utility,zhang2020low,wang2026adaptive,zhu2024trade,zhang2021learning} and full-vector~\cite{song2025adaptive,li2021sliding,cohen2020adaptive,cohen2022deepnp} coding paradigms have been comprehensively investigated in existing literature, establishing a foundational baseline for diverse network protocol designs. Building upon these insights, early optimization frameworks primarily focused on code rate adjustment through either static (offline)~\cite{garrido2019joint,guo2023network,zhang2020low,malak2019tiny,pandi2017pace,thai2014joint} or adaptive (online) mechanisms~\cite{huang2025utility,mao2021deep,wang2026adaptive,latzko2023analysing,cohen2022deepnp,zhang2021learning}. More specifically, offline approaches determine the code rate prior to transmission based on channel statistics, whereas online schemes adapt the code rate in response to instantaneous channel and network dynamics. Most recent schemes leverage deep learning techniques, such as Deep Reinforcement Learning~\cite{wang2026adaptive,latzko2023analysing}, to optimize coded redundancy by learning the relationship between observed network conditions and coding configurations. Owing to their data-driven nature, these methods can enable rapid decision-making and near real-time adaptation, albeit at the cost of increased algorithmic and computational complexity, which may limit their applicability in latency-sensitive and resource-constrained scenarios.

Subsequent studies highlighted the critical role of the coding window size in SW-RLNC~\cite{wowmom2023_karetsi,song2025adaptive,zhu2024trade,ali2023learning} — a pursuit conceptually parallel to generation size optimization in block-based RLNC frameworks~\cite{yeow2009minimizing,yang2012adaptive}. The coding window size has been identified as a key performance parameter that directly influences the trade-off between transmission efficiency, decoding delay, and computational complexity. However, recent efforts 
recognized that the code rate and the coding window size are strongly coupled and cannot be effectively optimized in isolation~\cite{ew2024_karetsi}. Hence, the research trajectory has evolved toward the joint optimization of both parameters spanning both block-based~\cite{claudel2025linc,dilanchian2025adjustable} and SW-RLNC~\cite{karetsi2026bayesian}, with emerging schemes demonstrating substantial performance gains over isolated parameter tuning. In this scope, methodologies vary significantly. For instance,~\cite{dilanchian2025adjustable} exploits lightweight, reactive heuristic functions to co-adjust the code rate and generation size at runtime, whereas~\cite{claudel2025linc} formulates a rigorous in-network mathematical model to jointly optimize block-based coding parameters directly within the network fabric.
Our previous work~\cite{karetsi2026bayesian} proposes a novel Bayesian Optimization framework for SW-RLNC, leveraging statistical surrogate models to actively and sample-efficiently search the parameter space.

\subsection{Optimizing the Decoding Process}\label{decoder-side-literature}
At the receiver, incoming packets are ingested into a \textit{Decoding Matrix}, denoted as $\mathbb{D}$, where packet erasures are resolved progressively via an on-the-fly decoding framework~\cite{icommlet2009_bioglio,iccw2009_heide,networking2011_heide}. This on-the-fly operational paradigm is necessitated by the sliding window architecture, wherein payloads arrive as a continuous stream rather than being partitioned into disjoint generations or blocks. 
Within this context, a prominent decoding approach is On-the-Fly Gaussian Elimination (OFGE)~\cite{icommlet2009_bioglio}, a variant of classical Gaussian elimination that combines partial forward elimination with a swap heuristic. Although originally introduced for Luby Transform (LT) codes, OFGE can be readily extended to SW-RLNC. A second widely adopted approach is the Kodo algorithm~\cite{iccw2009_heide,networking2011_heide}, which applies forward elimination and backward substitution to keep the decoding matrix in reduced row-echelon form at all times.

In practical deployments, $\mathbb{D}$ must be finite and restricted to packets that actively contribute to erasure resolution. 
This constraint motivates the definition of a \emph{decoding window}, denoted as $\mathcal{D}_{w}$, which comprises the set of packets that are still relevant for the decoding process. The maximum size of $\mathcal{D}_{w}$, designated as the \emph{decoding window size} ($DW$), serves as a critical design parameter that governs a fundamental trade-off between decoding performance and computational complexity. Specifically, an expanded $DW$ enhances the probability of packet recovery by preserving more information, whereas a restricted $DW$ minimizes decoding latency and computational overhead. Despite its architectural significance, existing literature frequently overlooks the optimization of $DW$ by assuming a symmetric configuration where $DW\!=\!CW$~\cite{zhu2024trade}. Under this condition, decoding window management becomes trivial because it structurally mirrors the transmitter's FIFO queue management; packets are evicted sequentially only when the cardinality of $\mathcal{D}_{w}$ exceeds $DW$. Prematurely evicting a packet in this regime creates an algebraic decoupling, wherein the decoder may encounter incoming coded packets that incorporate native payloads already purged from $\mathcal{D}_{w}$, thereby rendering the new coded packets useless. Consequently, to guarantee structural decodability, the system must satisfy the fundamental architectural constraint: $DW \geq CW$.

Operating under the condition $DW\!\!>\!\!CW$ introduces two primary design challenges: (i)~determining the optimal decoding window size $DW$, and (ii)~establishing dynamic policies for appending and evicting packets from the decoding window $\mathcal{D}_{w}$. 
Wunderlich et al.~\cite{wunderlich2017_caterpillar,wunderlich2017_wireless} were among the first to address the former challenge by evaluating the impact of $DW$ on decoding performance, albeit under a conventional FIFO window management policy. 
Their empirical findings indicate that $DW \approx 1.5CW$ serves as a practical operational setpoint. In parallel, Roca et al.~\cite{roca2017less} explored the relationship between latency budget constraints and the configuration of SW-RLNC coding parameters. Their results were largely consistent with those of~\cite{wunderlich2017_caterpillar} for burst-loss channels, while suggesting that a larger decoding window is preferable in memoryless channels, highlighting the dependence of the optimal $DW$ on the underlying loss characteristics.
While prior literature has initiated the discussion on decoding window management, no fully realized eviction policy has been proposed. Most notably, the heuristics discussed in~\cite{wunderlich2017_wireless} to identify potential eviction candidates suffer from key limitations. First, these heuristics are incomplete, leaving a non-trivial subset of viable candidate packets undetected. Second, they overlook the critical constraint of in-order packet delivery. Finally,~\cite{wunderlich2017_wireless} stops short of establishing formal, mathematically proven-correct eviction rules, leaving a significant gap between candidate identification and practical execution. 
Consequently, due to the absence of this formal, rule-based foundation, no unified eviction policy is proposed, implemented and/or evaluated under realistic network conditions.
Subsequently,~\cite{ew2024_karetsi} presented a preliminary mechanism to update the decoding window contents, although the underlying design rationale was neither extensively discussed nor evaluated. This observation underscores the need for a more systematic investigation of decoding window control mechanisms. To the best of our knowledge, a comprehensive framework for the structural management of $\mathcal{D}_{w}$ remains unexplored in prior literature. 

This work proposes
two novel decoding window management strategies aimed at optimizing the trade-off between decoding efficacy and computational overhead.
Crucially, active management of $\mathcal{D}_{w}$ becomes meaningful only in the $DW\!>\!CW$ regime. 
Under this condition, the decoder must execute a sophisticated management policy to govern window boundary updates---specifically determining which packets to retain, which to discard, and precisely when to release recovered payloads to the upper layers. 
These tactical decisions are vital, as they not only preserve the baseline decoding probability but also directly influence the in-order delivery delay and the instantaneous size of $\mathcal{D}_{w}$. The latter parameter, in turn, inherently governs the computational complexity of the decoding process.

\section{DYNAMIC DECODING WINDOW MANAGEMENT}\label{dynamic-matrix-management}
\subsection{System Model and Notation}\label{system-model}
We consider a single-channel wireless communication system comprising a transmitter (encoder) and a receiver (decoder). 
The transmitter implements 
SW-RLNC --- configured in either systematic or full-vector mode --- as a FEC mechanism, thereby operating without receiver-to-sender feedback. Each native packet $p$ is uniquely identified by a monotonically increasing sequence number, $sn^{p}$. Encoded payloads are generated via linear combinations of the native packets currently residing within the coding window $\mathcal{C}_{w}$ over the Galois Field $GF(2^8)$. A coded packet $c$ is said to \emph{contain} (or \emph{protect}) a native packet $p$ if $p$ is included in the linear combination that forms $c$, i.e., $c = c_{p}p+\ldots$, where $c_{p}$ is the non-zero coefficient of $p$ in the linear combination. Structurally, a coded packet $c$ is indexed by the sequence number corresponding to the first non-zero coefficient, i.e., the minimum packet id among the packets $c$ contains. 
The boundaries of $\mathcal{C}_{w}$ are defined by the minimum and maximum sequence numbers of its constituent native packets, $\mathcal{C}_{w}.sn_{\min}$ and $\mathcal{C}_{w}.sn_{\max}$, respectively, satisfying the constraint $\mathcal{C}_{w}.sn_{\max} - \mathcal{C}_{w}.sn_{\min} < CW$.

\begin{table}[t]
\centering
\caption{Summary of Main Notation}
\label{tab:notation}
\begin{tabular}{@{}cl@{}}
\toprule
\textbf{Symbol} & \textbf{Description}\\
\midrule
$\mathcal{C}_{w}$ & Encoder coding window \\
$\mathcal{D}_{w}$ & Decoder decoding window \\
$CW$ & Coding window size, i.e., maximum size of $\mathcal{C}_{w}$ \\
$DW$ & Decoding window size, i.e., maximum size of $\mathcal{D}_{w}$ \\
$sn^{q}$ & Sequence number of packet $q$ \\
$\mathcal{C}_{w}.sn_{\min}$ & Minimum sequence number in $\mathcal{C}_{w}$ \\
$\mathcal{C}_{w}.sn_{\max}$ & Maximum sequence number in $\mathcal{C}_{w}$ \\   
$\mathcal{D}_{w}.sn_{\min}$ & Minimum sequence number in $\mathcal{D}_{w}$ \\
$\mathcal{D}_{w}.sn_{\max}$ & Maximum sequence number in $\mathcal{D}_{w}$ \\   
$last\_delivered$ & Sequence number of the last packet delivered in-order \\
\bottomrule
\end{tabular}
\end{table}
At the receiver, packet erasures are progressively resolved via the Kodo decoding algorithm~\cite{iccw2009_heide,networking2011_heide}, utilizing a decoding window $\mathcal{D}_{w}$ and a decoding matrix $\mathbb{D}$ which is kept in reduced row echelon form at all times. 
Any other decoding algorithm with the reduced row echelon form property could also be utilized.
The operational boundaries of $\mathcal{D}_{w}$ and $\mathbb{D}$ are common and, analogous to $\mathcal{C}_{w}$, strictly defined by the sequence numbers of the chronologically oldest ($\mathcal{D}_{w}.sn_{\min}$) and most recent ($\mathcal{D}_{w}.sn_{\max}$) packets they currently retain. The size of both structures is constrained so that $\mathcal{D}_{w}.sn_{\max} - \mathcal{D}_{w}.sn_{\min} < DW$.

An incoming packet $x$, after being processed, is placed in $\mathcal{D}_{w}$ and a corresponding row is added to $\mathbb{D}$. The row is indexed by the sequence number of the packet $sn^{x}$ and contains the coefficients of the linear combination that forms the packet. A packet in $\mathcal{D}_{w}$ (and the corresponding row in $\mathbb{D}$) can be in one of three states:
\begin{itemize}[leftmargin=1.2em, itemindent=0pt]
    \item DECODED: The packet $x$ is a directly received native packet or a coded packet that has been reduced (decoded) to a native packet. From the decoding matrix perspective, the corresponding row $sn^{x}$ contains a single non-zero element in the $sn^{x}$ column, with all other elements being zero.
    \item PARTIALLY\_DECODED: The packet $x$ is actually a coded packet that has been partially recovered, but additional information is required for complete reconstruction, i.e., for transforming it into a DECODED state. 
    In the decoding matrix, a partially decoded packet is represented by row $sn^{x}$ that contains multiple non-zero elements, indicating that it is still entangled with other packets and cannot be fully resolved. More specifically, the first non-zero coefficient is placed at column $sn^{x}$ (the pivot) and means that $x$ contains the native packet $p_{v}$ with sequence number $sn^{p_{v}}=sn^{x}$, i.e, $x=c_{pv}p_{v}+\ldots$. Any other non-zero coefficient at column $sn^{q} > sn^{x}$ means that $x$ also contains the native packet $q$, i.e., $x=c_{pv}p_{v}+c_{q}q+\ldots$, where $c_{q}$ is the non-zero coefficient.
    \item MISSING: This is a native packet for which no information has been received or recovered and is currently absent from the decoding matrix. Such packets are indirectly identified through the sequence numbers of the received packets, which indicate gaps in the expected sequence. In the decoding matrix, a missing packet is represented by a row that contains only zero elements.
\end{itemize}

\noindent Table~\ref{tab:notation} summarizes the related notation.

\subsection{Motivation}\label{motivation}

Recent literature has demonstrated that decoupling $DW$ from $CW$ significantly optimizes the operational envelope of SW-RLNC~\cite{wunderlich2017_caterpillar,roca2017less,ew2024_karetsi}.
Specifically, configuring the receiver with an expanded decoding window ($DW > CW$) relaxes the strict performance bounds of traditional symmetric configurations ($DW\!=\!CW$), establishing a more favorable trade-off space across decoding efficiency, in-order delivery delay, and computational complexity. 

However, we argue that unlocking the full potential
of the asymmetric regime requires a highly nuanced
approach to decoding window management. 
When $DW = CW$, window regulation is elementary, adhering to a standard FIFO eviction policy that implicitly mirrors the transmitter's coding state. Conversely, when $DW > CW$, a standard FIFO policy becomes strictly sub-optimal as it inadvertently forces the operational window size to its maximum upper bound, $DW$. To circumvent this inflation, the receiver must execute sophisticated, state-aware eviction decisions that dynamically minimize the effective window footprint toward $CW$ while preserving the maximum attainable erasure recovery probability. As a result, an optimized window management strategy directly alleviates the computational overhead and delivery latency associated with large matrix operations by selectively purging uninformative components. This optimization yields two distinct architectural advantages:
\begin{itemize}[leftmargin=1.2em, itemindent=0pt]
    \item \emph{Resource Minimization}: Given an invariant target for decoding efficiency and in-order delivery delay, an advanced management policy can achieve the target while strictly minimizing the required resources at the decoder.
    \item \emph{Performance Scaling}: Under a fixed computational or latency budget, an efficient management scheme permits the system to scale to significantly larger $DW$ values, thereby extracting superior decoding efficiency that would otherwise be prohibited by hardware constraints. 
\end{itemize}
Consequently, developing precise, state-aware eviction strategies is an essential directive for high-throughput, low-latency sliding window architectures.

\subsection{Proposed Decoding Window Management Schemes}
A complete management policy must define the conditions for packet admission, delivery to the upper layers and eviction from $\mathcal{D}_{w}$ (and $\mathbb{D}$). Below, we delineate the three mechanisms that comprise the proposed \emph{Early Elimination of Decoded (EED)} and \emph{Early Elimination of Decoded and Missing (EEDM)} decoding window management schemes, starting with EED. The workflow of both EED and EEDM is illustrated in  Fig.~\ref{fig:eed_eedm_blocks}.
\begin{figure}[t]
    \centering
    \includegraphics[width=0.95\linewidth]{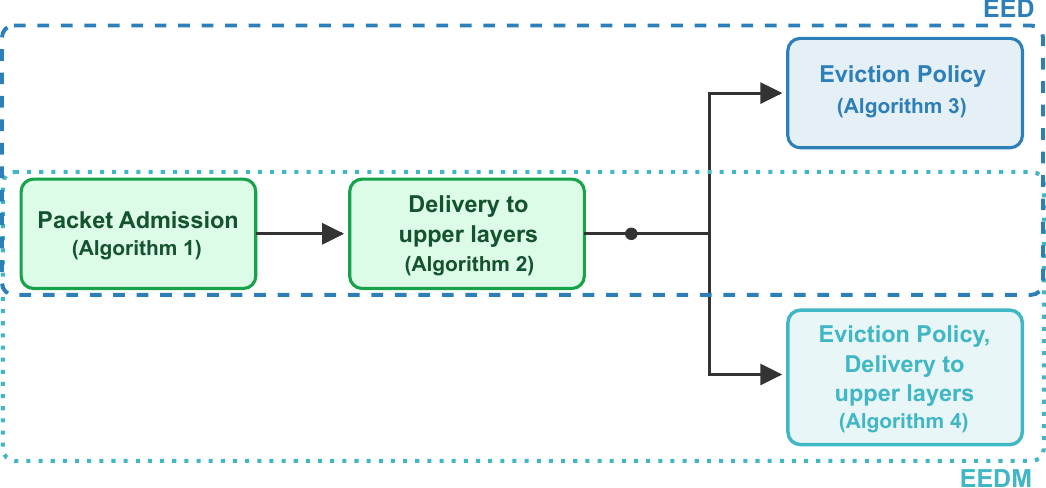}
    \caption{EED and EEDM workflows}
    \label{fig:eed_eedm_blocks}
\end{figure}

\subsubsection{EARLY ELIMINATION OF DECODED (\emph{EED})}

Adding new packets to the decoding window is straightforward (Algorithm~\ref{alg:packet_arrival}), as all incoming native packets are added to the decoding window upon arrival. This results in updating the upper boundary of the decoding window, $\mathcal{D}_{w}.sn_{\max}$, to the sequence number of the newly added packet. Coded packets are added to the decoding window only if the decoding algorithm renders them to be innovative, i.e., they contain new information that can contribute to the recovery of missing packets. If innovative, a coded packet will be processed and transformed by the on-the-fly decoding algorithm into either a DECODED or PARTIALLY\_DECODED packet, and it will be added to the decoding window with its sequence number set to the sequence number corresponding to the first non-zero coefficient it contains after processing.
\begin{algorithm}[t]
\caption{Packet Admission}\label{alg:packet_arrival}
\begin{algorithmic}[1]
\REQUIRE $\mathcal{D}_{w}$, incoming packet $p$
\IF{$p$ is a native packet}
    \STATE Add $p$ to $\mathcal{D}_{w}$ 
    \STATE $\mathcal{D}_{w}.sn_{\max} \gets sn^{p}$
\ELSIF{$p$ is a coded packet}
    \STATE Process $p$
    \IF{$p$ is innovative}
        \STATE $sn^{p} \gets$ index of first non-zero coefficient of $p$
        \STATE Add $p$ to $\mathcal{D}_{w}$
    \ELSE
        \STATE \textbf{discard} $p$
    \ENDIF
\ENDIF
\end{algorithmic}
\end{algorithm}

Delivery to upper layers is performed only for DECODED packets that are in-order. 
\begin{definition}
\label{def:in-order}
A DECODED packet $p$ is said to be \emph{in-order} iff all older packets in $\mathcal{D}_{w}$, i.e., those with lower sequence numbers, are also DECODED.
\end{definition}
\noindent Note that, delivering a packet to upper layers does not necessarily imply that it is removed from $\mathcal{D}_{w}$. This is because a DECODED packet may be useful for decoding future packets. The pseudocode of the basic routine for delivering packets to upper layers is provided in Algorithm~\ref{alg:packet_delivery}. This routine is executed immediately after adding packets to $\mathcal{D}_{w}$ (Algorithm~\ref{alg:packet_arrival}), since adding new coded packets may modify the in-order status of packets in $\mathcal{D}_{w}$. Also note that managing the delivery of packets requires an extra pointer to keep track of the last delivered packet, denoted as $last\_delivered$. This pointer is necessary because packets may be delivered to the upper layers before they are evicted from the decoding window.
\begin{algorithm}[!t]
\caption{Delivery to Upper Layers}
\label{alg:packet_delivery}
\begin{algorithmic}[1]
\REQUIRE $\mathcal{D}_{w}$, $last\_delivered$
\WHILE{$last\_delivered < \mathcal{D}_{w}.sn_{\max}$}
    \STATE $q \gets$ packet in $\mathcal{D}_{w}$ with sequence number $last\_delivered + 1$
    \IF{state($q$) = DECODED}
        \STATE Deliver $q$
        \STATE $last\_delivered \gets sn^{q}$
    \ELSE 
        \STATE \textbf{break} 
    \ENDIF
\ENDWHILE
\end{algorithmic}
\end{algorithm}

The more intricate aspect of decoding window management concerns the eviction of packets from $\mathcal{D}_{w}$. Under the baseline strategy commonly adopted in the literature, a packet — irrespective of its state — is evicted only once it reaches the last position of the decoding window, purely to enforce the size constraint $DW$. While we also utilize this strategy, in this work, we observe that DECODED packets can be safely evicted prior to reaching this boundary without compromising the decoding process and the in-order delivery of packets. 
Specifically, first we make the observation that:
\begin{lemma}\label{lm:not_useful}
    A packet $p$ cannot be used to decode any other packet in $\mathcal{D}_{w}$ if $\mathcal{D}_{w}.sn_{\max}-sn^{p}>CW$.
\end{lemma}
\begin{IEEEproof}
Since $\mathcal{D}_{w}$ is maintained in reduced row echelon form, every other row currently present in $\mathcal{D}_{w}$ has a zero coefficient at column $sn^{p}$. Moreover, since $\mathcal{D}_{w}.sn_{\max}-sn^{p}>CW$, $p$ lies outside the coding window $\mathcal{C}_{w}$. Thus, any coded packet subsequently entering $\mathcal{D}_{w}$ also has a zero coefficient at column $sn^{p}$. Consequently, every row, corresponding to a packet $q$, that is or will be present in $\mathcal{D}_{w}$ has a zero coefficient at column  $sn^{p}$, so using $p$ to eliminate this already-zero entry is a trivial operation that leaves the target row unchanged. Even if $p$ were, instead, used atypically to eliminate a coefficient of $q$'s at some other column, that would necessarily introduce a nonzero coefficient at column $sn^{p}$. This not only violates the reduce row echelon form invariant that $sn^{p}$ is $p$'s exclusive pivot column but also ties $q$'s decodability to column $sn^{p}$, which by assumption cannot be resolved. Hence, no use of $p$ — whether via its pivot column or otherwise — contributes usefully to the decoding of $q$.
\end{IEEEproof}
\noindent Then we can show that:
\begin{theorem}\label{thm:decoded}
A DECODED packet $p$ can be evicted from $\mathcal{D}_{w}$ without compromising the decoding-delivery process as soon as $\mathcal{D}_{w}.sn_{\max}-sn^{p}>CW$, provided that $p$ is in-order.
\end{theorem}
\begin{IEEEproof}
    By $\mathcal{D}_{w}.sn_{\max}-sn^{p}\!>\!CW$, $p$ lies outside $\mathcal{C}_{w}$ and is therefore, by Lemma~\ref{lm:not_useful}, not useful for decoding any other packet in $\mathcal{D}_{w}$. This alone, however, does not guarantee safe eviction, since eviction must also preserve correct in-order delivery. 
    To illustrate why $p$ must additionally be in-order, suppose that $p$ is evicted while not yet in-order. This means that there exists a packet $r$, older than $p$, i.e., $sn^{r}\!<\!sn^{p}$ that is either MISSING or PARTIALLY DECODED. If $r$ is subsequently decoded, this will force an out-of-order delivery, breaking the in-order delivery property. Requiring $p$ to be in-order avoids this: an in-order DECODED packet has, by definition, already been delivered together with every preceding packet. Its eviction therefore cannot compromise the correctness of subsequent deliveries.
\end{IEEEproof}
\begin{algorithm}[t]
\caption{\emph{EED} Eviction Strategy}
\label{alg:eed_decoder}
\begin{algorithmic}[1]
\REQUIRE $\mathcal{D}_{w}$, $CW$, $DW$
\WHILE{$\mathcal{D}_{w}.sn_{\max} - \mathcal{D}_{w}.sn_{\min} > CW$}
    \STATE $p \gets$ packet in $\mathcal{D}_{w}$ with sequence number $sn_{\min}$
    \IF{$\mathcal{D}_{w}.sn_{\max}\!-\!\mathcal{D}_{w}.sn_{\min}\!+\!1\!>\!DW$ \OR state($p$) = DECODED} \label{condition}
        \STATE Evict packet $p$ from $\mathcal{D}_{w}$
        \STATE $\mathcal{D}_{w}.sn_{\min} \gets \mathcal{D}_{w}.sn_{\min} + 1$
    \ELSE 
        \STATE \textbf{break} 
    \ENDIF
\ENDWHILE
\end{algorithmic}
\end{algorithm}
We refer to this proposed eviction scheme as EED, similar to the decoding window management strategy that utilizes it. Its pseudocode is provided in Algorithm~\ref{alg:eed_decoder}. Fig.~\ref{fig:eed_eedm_blocks} illustrates how the eviction policy blends with the other mechanisms to create the EED management scheme.
An important observation is that, when removing a packet from $\mathcal{D}_{w}$ (lines 3-6, Algorithm~\ref{alg:eed_decoder}), we do not need to forward it to the upper layers because this packet is in-order so it has already been delivered to upper layers by the delivery routine (Algorithm~\ref{alg:packet_delivery}).

\subsubsection{EARLY ELIMINATION OF DECODED AND MISSING (\emph{EEDM})}\label{EEDM}
A closer examination of \emph{EED} reveals that its eviction policy can be further enhanced by also considering the eviction of MISSING and PARTIALLY\_DECODED packets. Realizing this enhancement requires identifying two independent conditions under which such packets can be safely evicted from $\mathcal{D}_{w}$ without compromising the decoding-delivery process. We first derive the condition under which a MISSING packet can be safely evicted (Theorem~\ref{thm:missing}). We then extend the analysis to PARTIALLY\_DECODED packets, deriving the corresponding eviction condition in Theorem~\ref{thm:unrecoverable_partially_decoded}. We start by establishing the condition for MISSING packets.

\begin{theorem}\label{thm:missing}
A MISSING packet $p$ is unrecoverable if $\mathcal{D}_{w}.sn_{\max}\!-\!sn^{p}\!>\!CW$ and therefore can be safely evicted from $\mathcal{D}_{w}$.
\end{theorem}
\begin{IEEEproof}
Since $\mathcal{D}_{w}.sn_{\max}\!-\!sn^{p}\!>\!CW$, $p$ is outside the coding window, so no subsequently received coded packet $c$ will contain $p$, i.e., every incoming coded packet $c$ will have a zero coefficient at column $sn^{p}$. Thus, row $sn^{p}$ in $\mathbb{D}$ will remain all-zero, rendering $p$ unrecoverable. Thus it cannot affect the in-order delivery process. Furthermore, according to Lemma~\ref{lm:not_useful}, $p$ cannot be used to decode any other packet in $\mathcal{D}_{w}$. 
\end{IEEEproof}
\noindent Moving a step further, we can also identify unrecoverable PARTIALLY\_DECODED packets:
\begin{theorem}\label{thm:unrecoverable_partially_decoded}
A PARTIALLY\_DECODED packet in $\mathcal{D}_{w}$ that is older than an unrecoverable MISSING packet is also unrecoverable and can be safely evicted.
\end{theorem}
\begin{IEEEproof}
Let us consider a PARTIALLY\_DECODED packet $q$ in $\mathcal{D}_{w}$ that is older than an unrecoverable MISSING packet $p$. Since $q$ is still not fully decoded, row $sn^{q}$ in $\mathbb{D}$ contains at least two non-zero coefficients. The first is in column $sn^{q}$ (the pivot) and the second in a column $s'$.\\ 
\textbf{Case 1 ($s'\!<\!sn^{p}$)}: The non-zero coefficient at column $s'$ implies row $s'$ is missing the pivot, i.e., the packet $c$ with sequence number $s'$ is missing. This is because, if $c$ had already been received or decoded, the decoding algorithm would have immediately used the pivot at row $s'$ to clear column $s'$ in row $sn^{q}$ so that $\mathbb{D}$ is kept in reduced row echelon form. Moreover, the distance of $s'$ from the maximum sequence number in $\mathcal{D}_{w}$ is $\mathcal{D}_{w}.sn_{\max}\!-\!s'$ and since $s'\!<\!sn^{p}$, we have $\mathcal{D}_{w}.sn_{\max}\!-\!s'\!>\!\mathcal{D}_{w}.sn_{\max}\!-\!sn^{p}$. But $\mathcal{D}_{w}.sn_{\max}\!-\!sn^{p}\!>\! CW$ because $p$ is unrecoverable (Theorem~\ref{thm:missing}). Thus, $\mathcal{D}_{w}.sn_{\max}\!-\! s'\!>\!CW$. This means that $c$ is outside $\mathcal{C}_{w}$ and therefore no subsequent coded packet will contain $c$. As a result, there is no way to remove $c$ from $q$ and therefore $q$ is unrecoverable.\\
\textbf{Case 2 ($s' \geq sn^{p}$)}: In this case we know that $sn^{q}\!<\!sn^{p}$ because $q$ is older than $p$. 
Combining $s' \geq sn^{p}$ and $sn^{q} < sn^{p}$ yields the inequality chain $sn^{q} < sn^{p} \leq s'$. This indicates that the coefficient at column $sn^{p}$ is included in $q$.
Since $p$ is MISSING, row $sn^{p}$ is missing its pivot. So, there is no way that the decoding algorithm could have used it to reduce to zero the coefficient at column $sn^{p}$ in row $sn^{q}$\footnote{
While the probability of a zero coefficient in $\text{GF}(2^8)$ is non-zero ($1/256$), the joint probability of a PARTIALLY\_DECODED packet having a zero coefficient at column $sn^p$ while containing only two non-zero coefficients ($sn^q$ and $s'$) is negligible in practice. Alternatively, explicitly avoiding zero coefficients during coded packet creation completely eliminates the issue with negligible impact on linear independence.}. 
Furthermore, $\mathcal{D}_{w}.sn_{\max}-sn^{p}>CW$ (since $p$ is unrecoverable), therefore no subsequent coded packet will contain $p$. As a result, there is no way to remove $p$ from $q$ and therefore $q$ is unrecoverable.
Finally, since $sn^{q}<sn^{p}$ (older packet) it follows that $\mathcal{D}_{w}.sn_{\max}\!-\! sn^{q}\!>\!CW$ and according to Lemma~\ref{lm:not_useful}, $q$ cannot be used to decode any other packet in $\mathcal{D}_{w}$. So eviction is safe.
\end{IEEEproof}

\begin{algorithm}[t]
\caption{\emph{EEDM} Eviction Strategy \& Delivery to Upper Layers}
\label{alg:eedm_decoder}
\begin{algorithmic}[1]
\REQUIRE $\mathcal{D}_{w}$, $CW$, $DW$
\WHILE{$\mathcal{D}_{w}.sn_{\max} - \mathcal{D}_{w}.sn_{\min} > CW$}
    \STATE $p \gets$ packet in $\mathcal{D}_{w}$ with sequence number $\mathcal{D}_{w}.sn_{\min}$
    \IF {state($p$) = PARTIALLY\_DECODED}
        \STATE $sn_{nz}$ $\gets$ seq num of first non-zero coefficient after pivot in $p$
    \ENDIF
    \IF{$\big[(\mathcal{D}_{w}.sn_{\max}\!-\!\mathcal{D}_{w}.sn_{\min}\!+\!1\!>\!DW$ \\
    \OR state($p$) = DECODED \OR state($p$) = MISSING \\
    \OR $($state($p$) = PARTIALLY\_DECODED \AND $\mathcal{D}_{w}.sn_{\max}-sn_{nz}>CW)\big]$}
        \IF{state($p$) = DECODED \AND $sn^{p} > last\_delivered$}
            \STATE Deliver $p$
            \STATE $last\_delivered \gets sn^{p}$
        \ENDIF
        \STATE Evict packet $p$ from $\mathcal{D}_{w}$
        \STATE $\mathcal{D}_{w}.sn_{\min} \gets \mathcal{D}_{w}.sn_{\min} + 1$        
    \ELSE 
        \STATE \textbf{break} 
    \ENDIF
\ENDWHILE
\end{algorithmic}
\end{algorithm}
While Theorem~\ref{thm:unrecoverable_partially_decoded} effectively identifies which PARTIALLY\_DECODED packets could be evicted from $\mathcal{D}_{w}$, its implementation is not straightforward. This is because packets are examined from the oldest to the newest, thus, a PARTIALLY\_DECODED packet may be examined before the MISSING one that renders it unrecoverable. For a practical implementation, we utilize an efficient eviction rule:
\begin{corollary}
A PARTIALLY\_DECODED packet $q$ for which the left-most non-zero coefficient besides the pivot is in a column with sequence number $s_{m}$ such that $\mathcal{D}_{w}.sn_{max}\!-\!s_{m}\!>\!CW$ is unrecoverable and can be safely evicted from $\mathcal{D}_{w}$.
\end{corollary}
This corollary follows directly from Theorem~\ref{thm:unrecoverable_partially_decoded}, whose proof establishes that a PARTIALLY\_DECODED packet is unrecoverable if it has a nonzero coefficient at some column $s'$ with $\mathcal{D}_{w}.sn_{max}\!-\!s'\!>\!CW$. Since $s_{m}$ denotes the left-most nonzero coefficient of $q$ besides the pivot, it satisfies $s_{m} \le s'$ for any such column $s'$, and therefore $\mathcal{D}_{w}.sn_{max}-s_{m} \ge \mathcal{D}_{w}.sn_{max} - s' > CW$. Hence the hypothesis of the corollary, i.e. $\mathcal{D}_{w}.sn_{max}-s_{m} > CW$, is itself sufficient to establish unrecoverability, and safe eviction follows as in Theorem~\ref{thm:unrecoverable_partially_decoded}. The rule can serve as an immediate buffer dropping criterion and is highly elegant because it reduces an otherwise complex check down to a single, 
local O(1) test per row.

In general, \emph{EEDM} utilizes the same mechanisms as \emph{EED} for packet admission and delivery of packets to upper layers but replaces the eviction mechanism (Fig.~\ref{fig:eed_eedm_blocks}).
Theorems~\ref{thm:decoded},~\ref{thm:missing} and~\ref{thm:unrecoverable_partially_decoded} describe the rules that constitute the 
eviction scheme used in \emph{EEDM}. The pseudocode of this eviction scheme is provided in Algorithm~\ref{alg:eedm_decoder}.
An important observation is that, contrary to \emph{EED} (Algorithm~\ref{alg:eed_decoder}), here, when removing a packet from $\mathcal{D}_{w}$, we may need to first forward it to the upper layers (see lines 9-12, Algorithm~\ref{alg:eedm_decoder}). This is required when the status of the packet changes to "in-order" just after a previous deletion (e.g., a deletion of a MISSING or a PARTIALLY\_DECODED packet). 

Crucially, the proposed eviction criteria are predicated upon the presence of unrecoverable MISSING packets within the decoding matrix. This characteristic can be strategically exploited by executing Algorithm 4 strictly when such packets emerge, thereby guaranteeing that each row in the decoding matrix is evaluated at most once. Furthermore, because the individual eviction criteria require only $\mathcal{O}(1)$ operations, their computational overhead is trivial compared to the penalty of non-eviction, which would force the retention of eviction-eligible packets and their redundant inclusion in multiple subsequent decoding operations.
Consequently, in the ensuing analysis, the computational cost of executing Algorithm~\ref{alg:eedm_decoder} is considered negligible relative to the overall complexity of the decoding process.

\section{EVALUATION}\label{evaluation}

\subsection{Simulation Setup and Performance Metrics}
We evaluate the performance of the proposed decoding window management strategies through extensive simulations conducted using a custom time-slotted SW-RLNC module developed on top of the main API of Steinwurf’s PyErasure tool~\cite{pyerasure}. The simulation setup is built to reflect the system model introduced in Section~\ref{system-model}. The system is evaluated under a range of code rates $R$. To keep the parameter space realistic, we set $CW=60$ packets, a mid-range value commonly used in literature~\cite{tasdemir2021sparec} that provides a good balance between decoding efficiency and computational complexity. At the receiver side, the Kodo decoding algorithm~\cite{iccw2009_heide,networking2011_heide} is employed to resolve packet erasures across various $DW$ values. 

We compare the two proposed decoding window management strategies, namely \emph{EED} and \emph{EEDM},
against the conventional \textit{Evict-Oldest-on-Full} (\emph{EOF}) policy, which serves as the comparative benchmark. The \emph{EOF} policy is the state-of-the-art approach widely adopted in the literature~\cite{wunderlich2017_caterpillar,wunderlich2017_wireless}, where packets are evicted from the decoding window only when they reach the last position of the decoding window, regardless of their state. It is equivalent to Algorithm~\ref{alg:eed_decoder} 
except that it omits the second condition in line~\ref{condition}.

To capture the temporally correlated and bursty nature of wireless fading channels, we adopt a discrete-time Gilbert-Elliott (GE) channel framework~\cite{gilbert1960capacity}. The GE channel is modeled as a first-order, two-state Markov chain alternating between a Good state ($G$), characterized by error-free packet transmission, and a Bad state ($B$), wherein all packets are dropped. 
The state transition dynamics are governed by the probabilities $q_{GB}$ (transitioning from $G$ to $B$) and $q_{BG}$ (transitioning from $B$ to $G$). Accordingly, the stationary probability of the channel residing in the Bad state yields the average packet erasure rate $p_{\text{loss}}$, expressed as: $\pi_B = \frac{q_{GB}}{q_{GB} + q_{BG}} = p_{\text{loss}}$. Furthermore, the average error burst length, denoted as $\mathbb{E}[B]$, is structurally determined by these transition probabilities and is given by $\mathbb{E}[B] = \frac{1}{q_{BG}}$.

The simulation parameters are summarized in Table~\ref{tab:parameters}. Unless otherwise specified, the default values are used for all parameters. Each simulation is executed for a duration of $50000$ packet transmissions, and the results are averaged over 10 independent runs, which yields a worst-case $95\%$ confidence interval of $3\%$.
\begin{table}[t]
\centering
\caption{Simulation Parameters}
\label{tab:parameters}
\begin{tabular}{@{}lcc@{}}
\toprule
\textbf{Parameter} & \textbf{Range of Values} & \textbf{Default Value}\\
\midrule
$R$  & \(\{3/5,2/3,5/7,3/4,4/5\}\) & \(4/5\)\\
$CW$ & 60 & 60 \\
$DW$ & $CW$, 1.5$CW$, 2$CW$, 3$CW$ & 1.5$CW$\\   
$p_{loss}$      & $5\%, 10\%, 20\%$ & $10\%$\\
$E[B]$           & 4, 8, 12 slots & 8 slots\\
Transmission Mode & Systematic, Full Vector & Systematic \\
Packet Size & 1000 bytes & 1000 bytes \\
Simulation Duration & 50000 packets & 50000 packets \\
\bottomrule
\end{tabular}
\end{table}

We evaluate the performance of the three decoding window management policies using the following metrics:
\begin{itemize}[leftmargin=1.5em, itemindent=0pt]
    \item \emph{Delivery Ratio:} The proportion of successfully fully recovered 
    packets delivered to the destination relative to the total number of transmitted native packets. This metric quantifies the overall reliability and efficacy of the erasure correction process.
    \item \emph{Avg. In-Order Delivery Delay:} The mean time elapsed between the initial transmission of a native packet and its sequential, in-order delivery to the application layer. This metric characterizes the latency overhead induced by the decoding process, which is a critical performance factor for real-time applications. 
    \item \emph{Avg. Decoding Matrix Size:} The average number of packets retained within the decoding matrix. This metric serves as a measure of the decoder's memory footprint and hardware storage overhead.            
    \item \emph{Avg. Finite Field (FF) Operations:} The mean number of arithmetic operations (i.e., additions and multiplications over $GF(2^8)$) executed per decoding attempt. This metric assesses the computational complexity imposed by the decoder, which varies as a function of both the decoding window configuration and the management policy.
\end{itemize}

The evaluation is structured around three distinct experimental campaigns, each targeting a specific dimension of the decoding framework. The first experiment evaluates the resilience and effectiveness of the proposed strategies under varying channel conditions, from mild to severe, to assess their robustness and adaptability in real-world wireless environments.
The second experiment investigates the interplay between the management strategy and the decoding window size, highlighting the scaling of performance gains. Finally, the third experiment characterizes the performance impact of these management strategies across different transmission paradigms, specifically comparing the systematic and full-vector operational modes.

\subsection{Impact of Channel Conditions on Performance}
In this experiment we aim at capturing the performance of the proposed decoding window management strategies for different code rates and under a range of channel conditions. Specifically, we evaluate four channel configurations, summarized in Table~\ref{tab:ge_channel_settings}.
\begin{table}[b]
\centering
\caption{Gilbert--Elliott Channel Parameter Settings}
\label{tab:ge_channel_settings}
\begin{tabular}{@{}lcclcc@{}}
\toprule
Setting & $p_{loss}$ & $E[B]$ & Setting & $p_{loss}$ & $E[B]$\\
\midrule
(i) Mild            & 0.05 & 4  & (iii) Challenging & 0.2 & 12\\
(ii) Intermediate   & 0.10 & 8  & (iv) Burst-isolating & 0.1 & 12\\
\bottomrule
\end{tabular}
\end{table}
Settings (i)-(iii) jointly increase $p_{\text{loss}}$ and $E[B]$ to span a range of operating conditions from mild (Fig.~\ref{fig:expB-easychannel}) to challenging (Fig.~\ref{fig:expB-hardchannel}), with setting (ii) (Fig.~\ref{fig:expB-averagechannel}) representing an intermediate case. This joint scaling reflects the fact that, in many practical wireless and satellite links, higher average loss rates are empirically associated with longer fade/burst durations. Setting (iv) (Fig.~\ref{fig:expB-averageburstychannel}) is designed to isolate the effect of burst structure from that of average loss severity. By construction, it shares its loss rate with setting (ii) and its mean burst duration with setting (iii), enabling two complementary controlled comparisons: (ii) vs. (iv) attributes performance differences to burstiness alone at fixed $p_{\text{loss}}$, while (iii) vs. (iv) attributes differences to loss severity alone at fixed $E[B]$.

{
\begin{figure*}[t]
\centering
\subfloat[\label{subfig:expB-easychannel-delay}]{
    \includegraphics[width=0.325\textwidth]{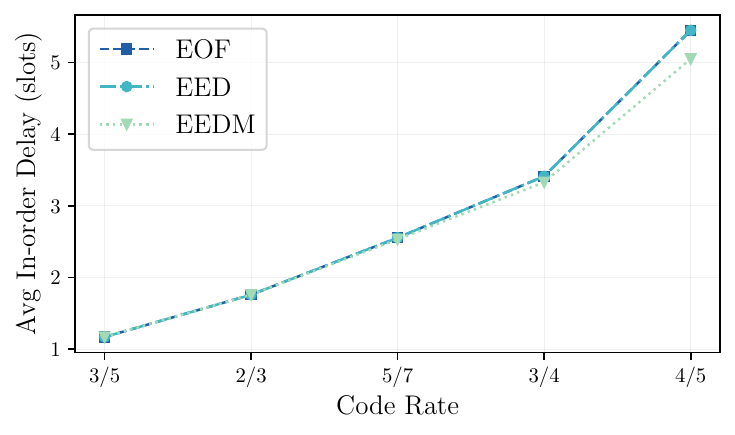}
}
\subfloat[\label{subfig:expB-easychannel-dm-size}]{
    \includegraphics[width=0.325\textwidth]{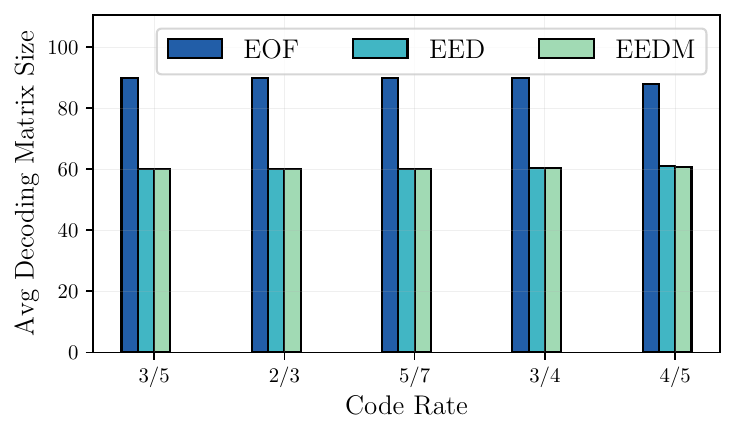}
}
\subfloat[\label{subfig:expB-easychannel-ff-ops}]{
    \includegraphics[width=0.325\textwidth]{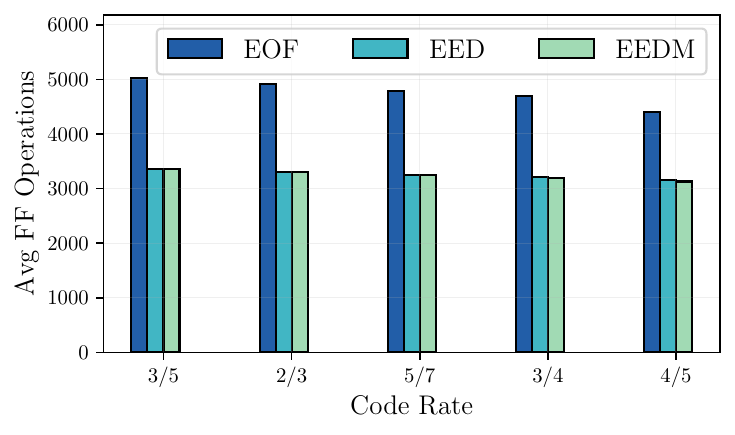}
}
\caption{Performance of \emph{EOF}, \emph{EED} and \emph{EEDM} for different code rates. Systematic coding, $p_{\text{loss}}=5\%$, $E[B]=4$, $DW=1.5CW$.}
\label{fig:expB-easychannel}
\end{figure*}
\begin{figure*}[t]
\centering
\subfloat[\label{subfig:expB-averagechannel-delay}]{
    \includegraphics[width=0.325\textwidth]{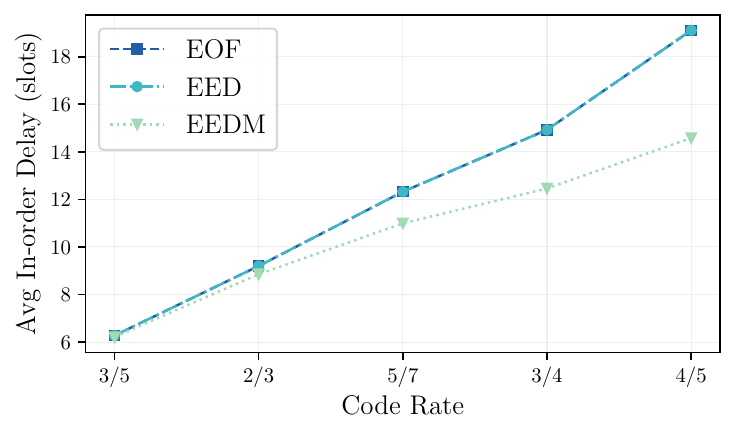}
}
\subfloat[\label{subfig:expB-averagechannel-dm-size}]{
    \includegraphics[width=0.325\textwidth]{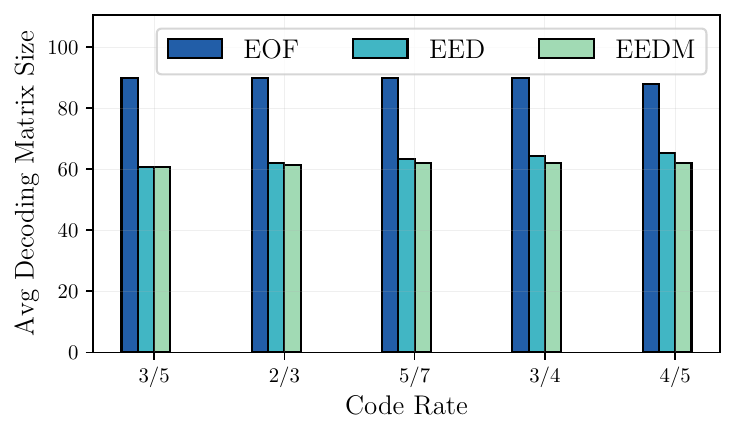}
}
\subfloat[\label{subfig:expB-averagechannel-ff-ops}]{
    \includegraphics[width=0.325\textwidth]{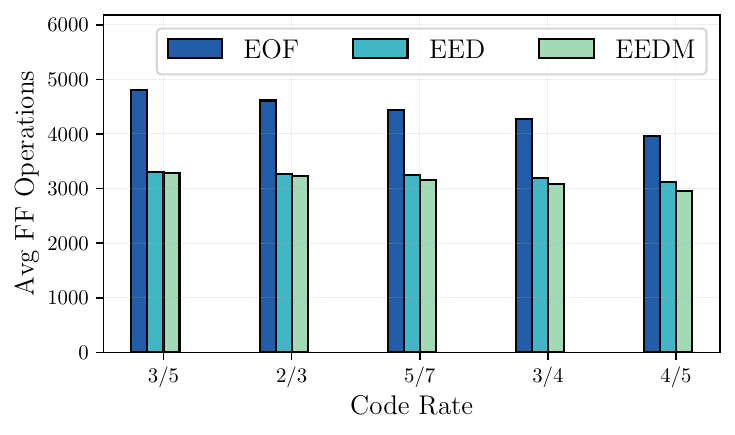}
}
\caption{Performance of \emph{EOF}, \emph{EED} and \emph{EEDM} for different code rates. Systematic coding, $p_{\text{loss}}=10\%$, $E[B]=8$, $DW=1.5CW$.}
\label{fig:expB-averagechannel}
\end{figure*}
\begin{figure*}[t]
\centering
\subfloat[\label{subfig:expB-hardchannel-delay}]{
    \includegraphics[width=0.325\textwidth]{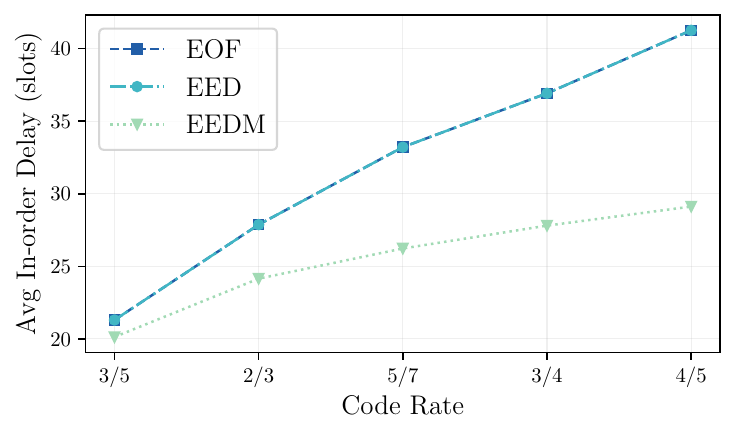}
}
\subfloat[\label{subfig:expB-hardchannel-dm-size}]{
    \includegraphics[width=0.325\textwidth]{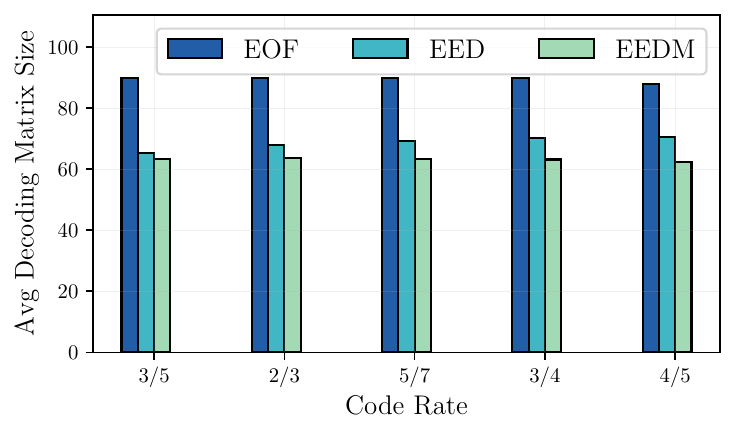}
}
\subfloat[\label{subfig:expB-hardchannel-ff-ops}]{
    \includegraphics[width=0.325\textwidth]{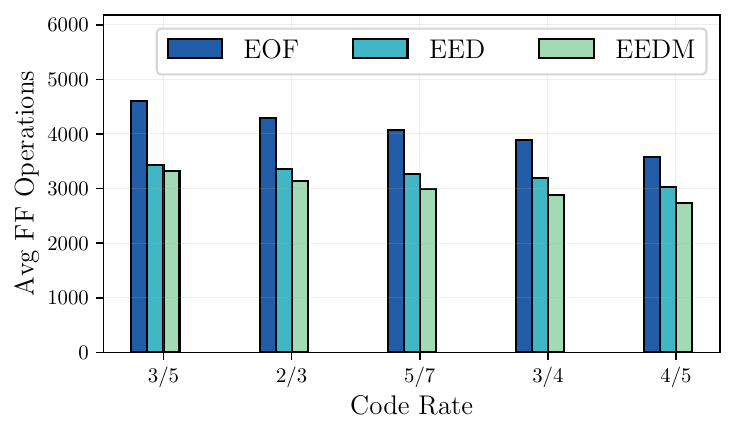}
}
\caption{Performance of \emph{EOF}, \emph{EED} and \emph{EEDM} for different code rates. Systematic coding, $p_{\text{loss}}=20\%$, $E[B]=12$, $DW=1.5CW$.}
\label{fig:expB-hardchannel}
\end{figure*}
\begin{figure*}[t]
\centering
\subfloat[\label{subfig:expB-averageburstychannel-delay}]{
    \includegraphics[width=0.325\textwidth]{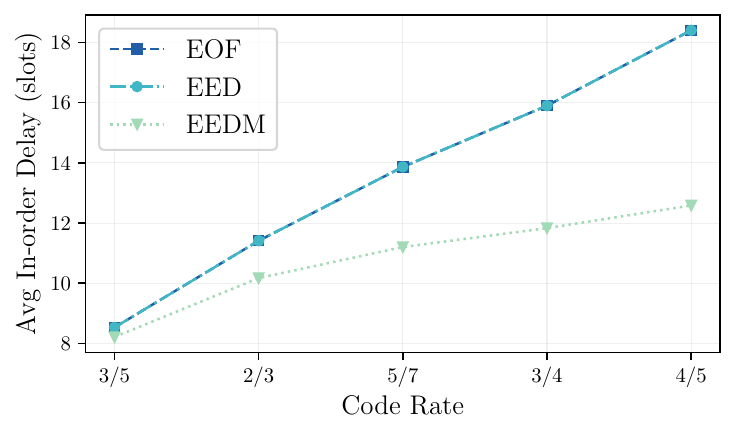}
}
\subfloat[\label{subfig:expB-averageburstychannel-dm-size}]{
    \includegraphics[width=0.325\textwidth]{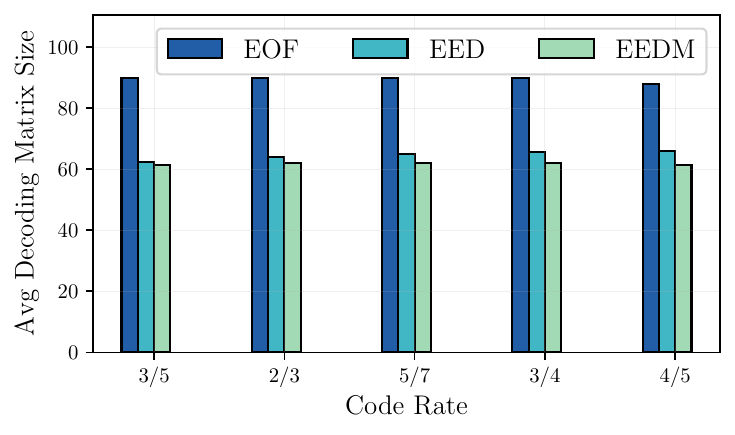}
}
\subfloat[\label{subfig:expB-averageburstychannel-ff-ops}]{
    \includegraphics[width=0.325\textwidth]{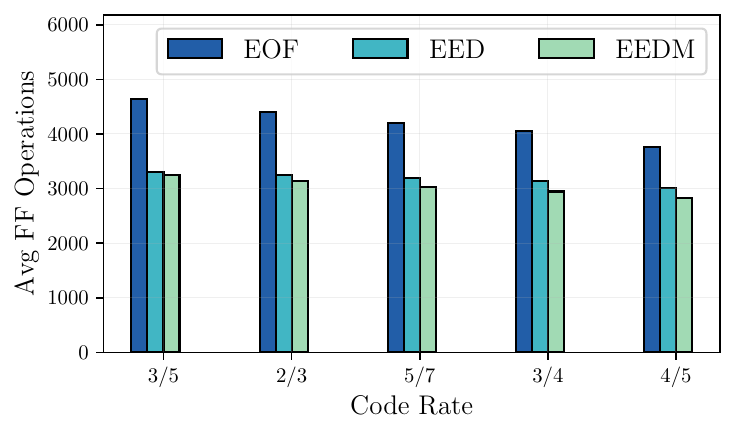}
}
\caption{Performance of \emph{EOF}, \emph{EED} and \emph{EEDM} for different code rates. Systematic coding, $p_{\text{loss}}=10\%$, $E[B]=12$, $DW=1.5CW$.}
\label{fig:expB-averageburstychannel}
\end{figure*}
}
Figures \ref{fig:expB-easychannel}-\ref{fig:expB-averageburstychannel} report the performance of the \emph{EOF} algorithm and the two proposed strategies, \emph{EED} and \emph{EEDM}, across the four channel configurations and for different code rates. Each subfigure corresponds to one performance metric: average in-order delivery delay, average decoding matrix size and average number of FF operations per decoding attempt. Across all considered channel settings and code rates, the three algorithms achieve statistically indistinguishable delivery ratios. This confirms that neither \emph{EED} nor \emph{EEDM} trades off decoding efficiency for the complexity and delay improvements discussed below; the proposed window management strategies operate strictly on packets that are either already delivered or provably non-recoverable, and therefore do not affect the receiver's ability to eventually recover all decodable packets. For brevity, the corresponding results are omitted.

Both \emph{EED} and \emph{EEDM} substantially reduce the average decoding matrix size relative to the \emph{EOF} algorithm (Fig.~\ref{subfig:expB-easychannel-dm-size}-~\ref{subfig:expB-averageburstychannel-dm-size}). In the \emph{EOF} algorithm, the matrix size is, by construction, always close to $DW$, since packets are retained until they exit the decoding window regardless of their decoding status. In contrast, under \emph{EED} and \emph{EEDM} the average matrix size closely tracks $CW$, with only a modest increase observed under more challenging channel conditions — that is, more severe loss rates (Fig.~\ref{fig:expB-easychannel} through Fig.~\ref{fig:expB-hardchannel}), longer burst durations (Fig.~\ref{fig:expB-averageburstychannel}), or higher code rates. Notably, even in these more demanding scenarios the increase over $CW$ remains limited, indicating that both proposed policies retain their effectiveness across the full range of tested conditions. The gap between \emph{EED} and \emph{EEDM} is comparatively small, though consistently observable. This is attributable to \emph{EEDM}'s more aggressive eviction policy, which additionally removes MISSING and PARTIALLY\_DECODED packets that \emph{EED} conservatively retains. Nonetheless, since both policies already operate close to the theoretical minimum matrix size $CW$ (i.e., no packet is removed before reaching position $CW+1$), the additional margin available to \emph{EEDM} is inherently limited. The gap between the two widens under more challenging channel configurations, where a larger number of MISSING and PARTIALLY\_DECODED packets fall within the $[CW+1,DW]$ range of the decoding matrix, thereby giving \emph{EEDM}'s more aggressive eviction criterion a broader operating margin.

The trends observed for average FF operations per decoding attempt largely mirror those of the decoding matrix size (Fig.~\ref{subfig:expB-easychannel-ff-ops}-\ref{subfig:expB-averageburstychannel-ff-ops}). Both \emph{EED} and \emph{EEDM} outperform the \emph{EOF} algorithm, which is expected; since incoming packets are processed against a smaller number of matrix rows, the number of required FF operations decreases accordingly. The relative performance gain of \emph{EED} and \emph{EEDM} over the \emph{EOF} algorithm narrows — while remaining clearly noticeable — under more challenging channel conditions. This is not attributable to a degradation of \emph{EED}/\emph{EEDM} performance, but rather to an artifact of the \emph{EOF} algorithm's behavior; under harsher channel conditions, the \emph{EOF} algorithm's decoding matrix becomes increasingly populated by MISSING packets, whose coefficient rows are entirely zero and therefore require no operations upon interaction with incoming packets, which genuinely reduces the \emph{EOF} algorithm's operation count. This reduction, however, carries no benefit in terms of decoding outcome. Instead, the narrowing gap in raw operations therefore reflects the \emph{EOF} algorithm processing a larger share of uninformative rows.
Interestingly, the gap between \emph{EED} and \emph{EEDM} exhibits the opposite trend: \emph{EEDM}'s advantage over \emph{EED} slightly increases in more challenging scenarios. This follows directly from the decoding matrix size results discussed above — \emph{EEDM}'s more aggressive eviction of MISSING and PARTIALLY\_DECODED packets yields a comparatively smaller decoding matrix precisely in the scenarios where such packets would otherwise accumulate, translating directly into a further reduction of FF operations relative to \emph{EED}.

\emph{EED} achieves in-order delay performance statistically identical to the \emph{EOF} algorithm (Fig.~\ref{subfig:expB-easychannel-delay}-\ref{subfig:expB-averageburstychannel-delay}). This behavior is expected by design because \emph{EED} evicts only DECODED packets that lie outside the coding window and are already in-order; such packets, by definition of being in-order, have already been delivered to the upper layers prior to eviction. \emph{EED}'s eviction policy therefore has no bearing on delivery timing and serves purely to reduce decoding complexity, consistent with the matrix size and FF operation results above. \emph{EEDM}, in contrast, achieves substantial in-order delay reductions relative to both the \emph{EOF} algorithm and \emph{EED}. This improvement arises because \emph{EEDM}'s eviction of MISSING and PARTIALLY\_DECODED packets reorders the remaining DECODED packets within the window, which can change their status to in-order and thereby accelerate their delivery to the upper layers — an effect unavailable to \emph{EED}, whose eviction criterion does not act on non-decoded packet states. As with the previous metrics, the magnitude of \emph{EEDM}'s delay improvement scales with channel severity: gains are minimal under mild channel conditions and increase substantially under more challenging ones. This follows from two compounding effects under harsher conditions. First, the \emph{EOF} and \emph{EED} algorithms accumulate longer queues of out-of-order DECODED packets within the decoding window, meaning \emph{EEDM}'s reordering mechanism affects a larger number of packets. Second, the per-packet delay reduction is itself larger, since an out-of-order DECODED packet subject to \emph{EOF} or \emph{EED} policies remains queued for a proportionally longer duration before its blocking predecessors are resolved.

As a final note, across all four metrics, both $p_{\text{loss}}=10\%$ and $E[B]=8$ are shown to influence performance, with $p_{\text{loss}}=10\%$ exerting the more pronounced effect — consistent with its role as the dominant driver of channel severity in the Gilbert-Elliott model. In all tested configurations, \emph{EED} and \emph{EEDM} deliver significant performance improvements over the \emph{EOF} algorithm without any loss in delivery ratio. \emph{EED}'s benefits are confined to a marked reduction in decoding matrix size and FF operations, i.e., computational complexity, while leaving in-order delay unaffected. \emph{EEDM} extends these complexity gains further and additionally achieves substantial in-order delay reductions, at the cost of a more aggressive (though still simple to implement) eviction policy.

\subsection{The Effect of Decoding Window Management and Size}
\begin{figure*}[t]
\centering
\subfloat[\label{subfig:expA-delivery}]{
    \includegraphics[width=0.45\textwidth]{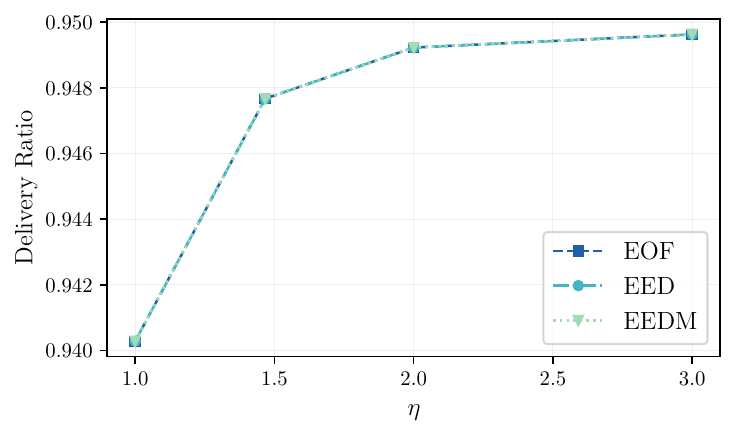}
}
\hfill
\subfloat[\label{subfig:expA-delay}]{
    \includegraphics[width=0.45\textwidth]{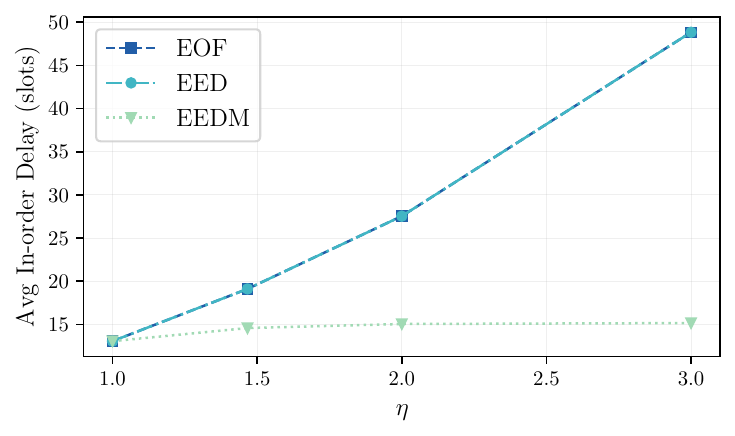}
}
\vspace{-0.8em}
\subfloat[\label{subfig:expA-dm-size}]{
    \includegraphics[width=0.45\textwidth]{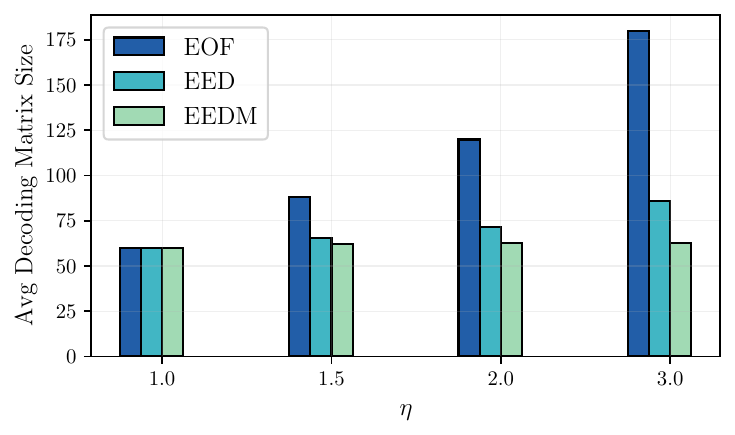}
}
\hfill
\subfloat[\label{subfig:expA-ff-ops}]{
    \includegraphics[width=0.45\textwidth]{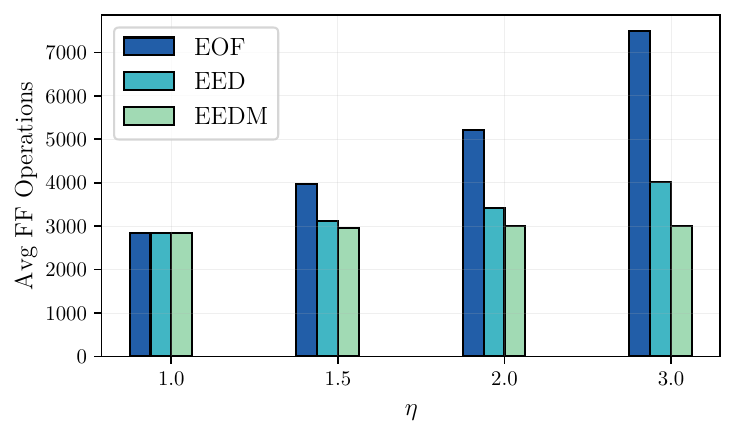}
}
\caption{Performance of \emph{EOF}, \emph{EED} and \emph{EEDM} for different $\eta=DW/CW$ values (Systematic coding, $p_{\text{loss}}=10\%$, $E[B]=8$, $R=4/5$).}
\label{fig:expA}
\end{figure*}
The second experiment investigates the sensitivity of the \emph{EOF}, \emph{EED}, and \emph{EEDM} algorithms to $DW$, parametrized as $\eta = DW/CW \in [1,3]$. This investigation is motivated by the observation that the $DW>CW$ regime affords the decoder additional flexibility to recover packets and may therefore improve decoding efficiency. Fig.~\ref{subfig:expA-delivery} reports the delivery ratio, ~\ref{subfig:expA-delay} the average in-order delay, ~\ref{subfig:expA-dm-size} the average decoding matrix size, and ~\ref{subfig:expA-ff-ops} the average FF operations as functions of $\eta$. 

As expected, all three algorithms exhibit identical performance at $\eta=1$ across all four metrics, since \emph{EED} and \emph{EEDM}'s eviction criteria only diverge from the \emph{EOF} algorithm's behavior once $DW>CW$. At $\eta=1$, the operational range $[CW+1,DW]$ in which the proposed policies act is empty by construction, and therefore all algorithms coincide.
Fig.~\ref{subfig:expA-delivery} confirms, in a second and independent setting, the theoretical guarantee that neither \emph{EED} nor \emph{EEDM} impairs decoding performance; all three algorithms achieve identical delivery ratios across the full range of $\eta$. The results also substantiate the practical value of the $DW>CW$ regime itself, as delivery ratio is shown to increase monotonically with $\eta$ for all algorithms. This confirms that enlarging $DW$ relative to $CW$ provides a genuine mechanism for improving decoding efficiency, independently of which eviction policy is employed.

Both \emph{EED} and \emph{EEDM} continue to yield substantial reductions in decoding matrix size and FF operations relative to the \emph{EOF} algorithm, and these improvements grow with $\eta$ This trend is a direct consequence of the eviction mechanism underlying both policies; as $\eta$ increases, the operational range $[CW+1,DW]$ over which \emph{EED} and \emph{EEDM} identify and remove non-essential packets widens correspondingly, creating more opportunities for eviction and consequently amplifying the relative gains over the \emph{EOF} algorithm, whose matrix size scales directly with $DW$ regardless of $\eta$.
The performance gap between \emph{EED} and \emph{EEDM} likewise widens with increasing $\eta$. This follows naturally from \emph{EEDM}'s more aggressive eviction criterion, which, as the operational range over which it is applied grows, accumulates a proportionally larger advantage over \emph{EED}'s more conservative policy. A particularly notable finding is the stability of \emph{EEDM}'s performance across the tested range of $\eta$. Even at $DW=3$, \emph{EEDM} exhibits only marginal degradation in both decoding matrix size and FF operations relative to its performance at lower $\eta$ values. This near-flat scaling behavior substantially improves the achievable trade-off between decoding efficiency and computational complexity, and in turn increases the practical attractiveness of adopting larger $DW$ values — a regime that would otherwise impose a computational and memory burden that scales directly with window size under the \emph{EOF} algorithm.

Consistent with the results of the first experiment, \emph{EED} and the \emph{EOF} algorithm exhibit statistically indistinguishable in-order delay across all tested values of $\eta$, reaffirming that \emph{EED}'s eviction policy — restricted to already-delivered, in-order DECODED packets — has no bearing on delivery timing regardless of the decoding window size. \emph{EEDM}, by contrast, achieves substantial and increasing in-order delay reductions as $\eta$ grows, with the improvement reaching approximately $60\%$ relative to \emph{EOF} and \emph{EED} at $\eta=3$. This is consistent with the mechanism identified in the first experiment; as $\eta$ increases, a larger number of MISSING and PARTIALLY\_DECODED packets populate the $[CW+1,DW]$ range under \emph{EOF} and \emph{EED}, delaying the in-order status, and thus, the delivery of subsequent DECODED packets; \emph{EEDM}'s eviction of these blocking packets directly alleviates this effect, with the benefit compounding as the window grows. Notably, \emph{EEDM}'s in-order delay exhibits remarkable stability across $\eta$, with only a minimal increase observed even at $DW=3CW$. This indicates that \emph{EEDM} effectively decouples in-order delay from $DW$, in stark contrast to the \emph{EOF} algorithm and \emph{EED}, whose delay grows with $\eta$.

Taken together, these results demonstrate that \emph{EEDM} substantially improves the trade-off between decoding efficiency, in-order delay, and computational complexity as $DW$ is enlarged relative to $CW$. While increasing $DW$ under the \emph{EOF} algorithm yields improved decoding efficiency only at the cost of increasing delay and computational burden, \emph{EEDM} largely decouples these costs from $DW$, keeping both delay and computational overhead close to their $\eta=1$ levels. Equivalently, given a fixed computational or latency budget, \emph{EEDM} makes it feasible to employ a substantially larger $DW$, resulting in higher decoding efficiency than would otherwise be attainable.

\subsection{Transmission Mode and Decoding Management}
\begin{figure*}[t]
\centering
\subfloat[\label{subfig:expC-delivery}]{
    \includegraphics[width=0.48\textwidth]{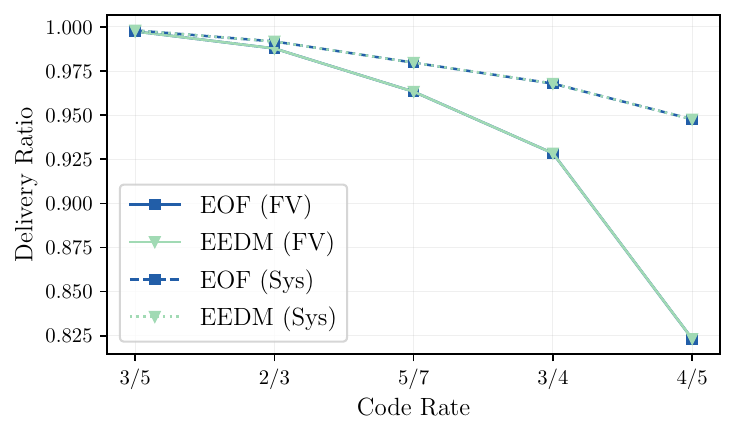}
}
\hfill
\subfloat[\label{subfig:expC-delay}]{
    \includegraphics[width=0.48\textwidth]{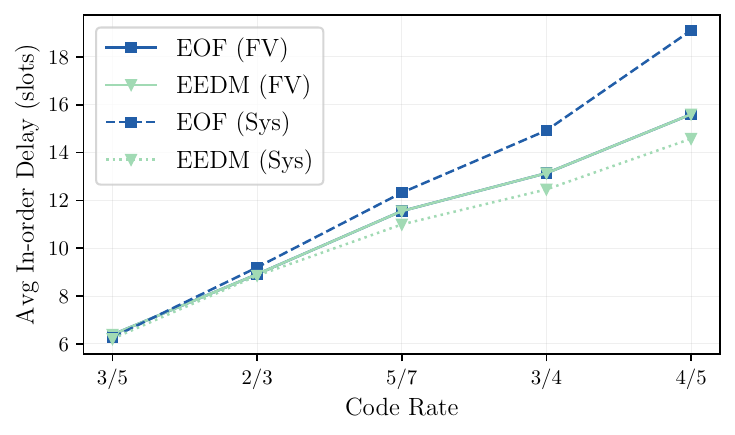}
}
\vspace{-0.8em}
\subfloat[\label{subfig:expC-dm-size}]{
    \includegraphics[width=0.48\textwidth]{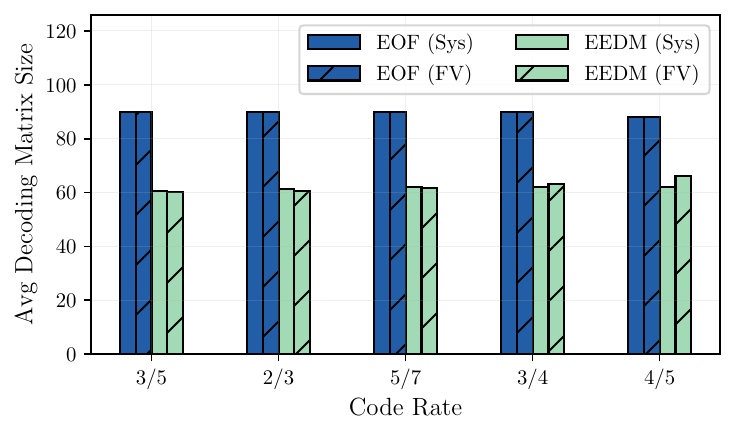}
}
\hfill
\subfloat[\label{subfig:expC-ff-ops}]{
    \includegraphics[width=0.48\textwidth]{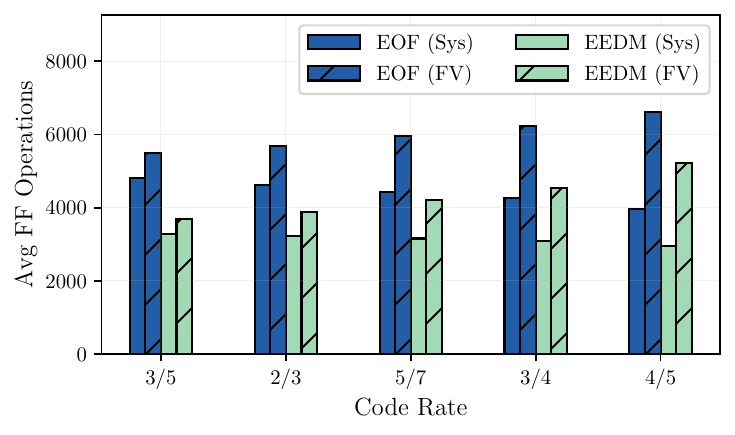}
}
\caption{Performance of \emph{EOF} and \emph{EEDM} with Full Vector (FV) and Systematic (Sys) coding for different code rates ($p_{\text{loss}}\!=\!10\%$, $E[B]\!=\!8$, $DW\!=\!1.5CW$).}
\label{fig:expC}
\end{figure*}
The preceding experiments consider systematic coding, which is the natural default for delay- and complexity-sensitive applications due to the immediate usability of native packets upon arrival. Nevertheless, there are practical scenarios in which the use of full-vector (FV) coding may be preferable, or even imperative, despite its higher baseline complexity. Examples include scenarios in which design considerations requires packet interchangeability, e.g., multi-hop, multi-path propagation, or payload confidentiality and obfuscation. Given these practical cases, it is important to establish whether the proposed window management strategies retain their benefits under FV coding, and, if so, to characterize how those benefits manifest differently relative to the systematic case. Fig.~\ref{fig:expC} addresses this question, comparing the \emph{EOF} and \emph{EEDM} algorithms under both coding modes across a range of code rates.

Fig.~\ref{subfig:expC-delay} shows that, while \emph{EEDM}'s in-order delay advantage over the \emph{EOF} algorithm is well preserved under systematic coding — consistent with the results of the previous experiments — this advantage nearly vanishes under FV coding, where \emph{EEDM} and \emph{EOF} exhibit almost identical in-order delay. 
This behavior is a direct consequence of the in-order-friendly property of FV coding. Because every transmitted packet is coded, there is no risk of receiving a native, newer packet out of turn. Instead, the decoder processes the matrix into reduced row echelon form, naturally resolving and releasing packets in chronological order from oldest to newest.
In other words, there are no native packets, that could arrive and be delivered out of order, which is precisely the mechanism through which \emph{EEDM} accelerates delivery under systematic coding (Section~\ref{EEDM}). With no out-of-order DECODED packets for \emph{EEDM}'s eviction mechanism to act upon, its delay advantage over \emph{EOF} is lost in FV mode.
A second, related observation concerns the low absolute in-order delay values obtained under FV coding, visible in Fig.~\ref{subfig:expC-delay} at high code rates. This apparent advantage is a statistical artifact rather than a genuine performance benefit, and must be interpreted jointly with Fig.~\ref{subfig:expC-delivery}. As code rate increases (i.e., redundancy decreases), FV coding exhibits a comparatively steeper decline in delivery ratio. Since undelivered (non-decoded) packets do not contribute to the in-order delay statistic, a lower delivery ratio mechanically biases the reported average delay downward — packets that would have contributed the largest delays are precisely those that fail to be decoded at all, and are excluded from the average by construction. The reduced in-order delay of FV coding at high code rates therefore reflects reduced decoding efficiency rather than improved delivery timeliness, and should not be read as an advantage of FV coding in this regime.

The in-order decoding property of FV coding also explains the modest additional advantage observed for \emph{EEDM} under FV relative to Sys mode in Fig.~\ref{subfig:expC-dm-size}, where both substantially outperform the \emph{EOF} algorithm. Because FV coding decodes packets strictly in order, a larger share of DECODED packets within the window are also in-order at any given time, compared to systematic coding. Recall that in-order DECODED packets are eligible for early eviction under \emph{EEDM}; the higher incidence of such packets under FV coding therefore affords \emph{EEDM} more opportunities for eviction, yielding a marginally smaller average decoding matrix relative to Sys mode. This trend reverses only at the highest tested code rate ($R=4/5$), where the reduced redundancy causes FV decoding efficiency to collapse (Fig.~\ref{subfig:expC-delivery}). With substantially fewer packets reaching DECODED status under FV in this regime, the pool of in-order-eligible packets available for \emph{EEDM}'s eviction shrinks correspondingly, and the matrix-size advantage of FV over Sys disappears.

As expected, FV coding incurs a higher computational cost than systematic coding under both algorithms (Fig.~\ref{subfig:expC-ff-ops}), since FV packets lack the single-nonzero-coefficient structure of native packets and therefore require full elimination against every row of the active decoding matrix. Nevertheless, \emph{EEDM} reduces FF operations relative to \emph{EOF} under both coding modes. This reduction stems from the same underlying mechanism identified in the earlier experiments. By maintaining a smaller active decoding matrix, \emph{EEDM} ensures that each incoming packet interacts with fewer rows, reducing the number of operations performed irrespective of whether those rows originate from native or coded packets.

These results indicate that, in scenarios where the use of FV coding is imperative, \emph{EEDM} continues to provide substantial computational benefits despite the loss of its delay advantage in this mode. This is particularly significant because FV coding typically requires higher redundancy than systematic coding to achieve comparable decoding efficiency (Fig.~\ref{subfig:expC-delivery}), which in turn compounds its inherent computational disadvantage. The magnitude of \emph{EEDM}'s benefit is sufficient to offset this compounded cost. As shown in Fig.~\ref{subfig:expC-ff-ops}, the FF operation count of \emph{EEDM} under FV coding at $R=2/3$ is comparable to that of the \emph{EOF} algorithm under systematic coding at $R=4/5$, while the former achieves a markedly higher delivery ratio (Fig.~\ref{subfig:expC-delivery}). In other words, \emph{EEDM} allows FV coding to be operated at a lower code rate — and hence with meaningfully higher decoding efficiency — at a computational cost no greater than an unoptimized systematic scheme operating at a higher, less reliable code rate. This result substantially broadens the practical applicability of FV coding in complexity-constrained deployments where it would otherwise be considered prohibitive.

\section{CONCLUSION}\label{conclusion}
This paper investigated decoding window management as a design dimension for SW-RLNC, complementing the encoder-side tuning that has dominated prior work. We derived conditions under which packets can be safely removed from the decoding window without compromising decodability or in-order delivery, and used them to design two lightweight management schemes, namely \emph{EED} and \emph{EEDM}, proving that both preserve decodability by construction. Simulation results across diverse channel conditions and code rates confirm this guarantee empirically, while showing that both schemes significantly reduce computational complexity, while \emph{EEDM} additionally achieves substantial reductions in in-order delivery delay, with benefits that grow in more challenging channel conditions and larger decoding windows.

These performance gains can be leveraged via two complementary approaches. First, for a fixed decoding window size, the proposed schemes reduce the computational complexity and latency required to achieve a target decoding efficiency. Second, given a constrained complexity or delay budget, they permit the utilization of larger decoding windows, thereby yielding higher decoding efficiency that would otherwise be unattainable. This flexibility renders the proposed framework highly attractive for latency- and resource-constrained deployments, such as URLLC and satellite networks. However, fully maximizing these practical gains requires a joint optimization framework. Specifically, the decoding window management policy, the coding window size, and the available receiver resources must be jointly optimized and dynamically adapted in real time to time-varying channel conditions. Orchestrating this joint optimization within a practical framework remains a subject of future work. Additionally, future investigations will examine the design and impact of more elaborate lossless or relaxed packet eviction rules to enhance the feasibility of SW-RLNC in practical deployments featuring resource-constrained hardware.

\section{Acknowledgments}
The research work was supported by the Hellenic Foundation for Research and Innovation (HFRI) under the 5th Call for HFRI PhD Fellowships (Fellowship Number: 19139).

\bibliographystyle{IEEEtran}
\bibliography{IEEEabrv, references}

\vfill

\end{document}